\documentclass{article}
\usepackage{amsmath}
\usepackage{amssymb}
\usepackage{amsfonts}
\usepackage{hyperref}

\newtheorem{theorem}{Theorem}

\newtheorem{example}[theorem]{Example}

\newtheorem{proposition}[theorem]{Proposition}
\newtheorem{remark}[theorem]{Remark}

\newenvironment{proof}[1][Proof]{\noindent\textbf{#1.} }{\ \rule{0.5em}{0.5em}}
\numberwithin{theorem}{section}
\numberwithin{equation}{section}

\input{tcilatex}
\begin{document}

\title{An anisotropic geometrical approach for non-relativistic extended
dynamics}
\author{M. Neagu\footnotemark[1]\;, A. Oan\u{a}\footnotemark[1] \; and V.M.
Red'kov\footnotemark[2] }
\date{}
\maketitle

\begin{abstract}
In this paper we present the distinguished (d-) Riemannian geometry (in the
sense of nonlinear connection, Cartan canonical linear connection, together
with its d-torsions and d-curvatures) for a possible Lagrangian inspired by
optics in non-uniform media. 
The corresponding equations of motion are also exposed, and some particular
solutions are given. For instance, we obtain as geodesic trajectories some
circular helices (depending on an angular velocity $\omega $), certain
circles situated in some planes (ones are parallel with $xOy $, and other
ones are orthogonal on $xOy$), or some straight lines which are parallel
with the axis $Oz$. All these geometrical geodesics are very specific
because they are completely determined by the non-constant index of
refraction $n(x).$
\end{abstract}

\thanks{}
\thanks{}


\footnotetext[1]{%
Department of Mathematics and Informatics, University Transilvania of Bra%
\c{s}ov, 50 Iuliu Maniu Blvd., 500091 Bra\c{s}ov, Romania, E-mails:
mircea.neagu@unitbv.ro, alexandru.oana@unitbv.ro} \footnotetext[2]{%
Institute of Physics, National Academy of Sciences of Belarus, Minsk,
Belarus, E-mail: redkov@dragon.bas-net.by} \setcounter{footnote}{3}



\section{Introduction}



In the geometrical optics \cite{Landau-Lif-1}, a central role is played by
the Synge-Beil metric \cite{Balan-Synge, Mir-Kaw, Neagu Optics TM}%
\begin{equation}
g_{\alpha \beta }(x,y)=\varphi _{\alpha \beta }(x)+\mathbf{\gamma }%
^{2}y_{\alpha }y_{\beta },  \label{Synge-Beil}
\end{equation}%
where $\mathbf{\gamma }(x)\geq 0$ is a positive smooth function on the
space-time $M^{4}$, and $\varphi _{\alpha \beta }(x)$ is a pseudo-Riemannian
metric on $M^{4}$. One assumes that the four-dimensional manifold $M^{4}$
(which is connected and simply connected) is endowed with the local
coordinates%
\begin{equation*}
\left( x^{\alpha }\right) _{\alpha =\overline{0,3}}=\left(
x^{0}=t,x^{1},x^{2},x^{3}\right) ;
\end{equation*}%
for simplicity we use the system of units where the light velocity is $c=1$.
Obviously, the following rule holds good: $y_{\alpha }=\varphi _{\alpha \mu
}y^{\mu }$. Because the components of $\varphi _{\alpha \beta }(x)$ are
dimensionless, the same are $\mathbf{\gamma }y_{\alpha }$: 
\begin{equation*}
\lbrack \varphi _{\alpha \beta }(x)]=1,\;[\mathbf{\gamma }y_{\alpha }]=1.
\end{equation*}

In such a context, let us restrict our geometric-physical study to the
Euclidean manifold $\mathcal{E}^{3}=\left( \Sigma ^{3},\delta _{ij}\right) $
which has the local coordinates $(x):=\left( x^{i}\right) _{i=\overline{1,3}%
}.$ It follows that the corresponding tangent bundle $T\Sigma ^{3}$ has the
dimension equal to six, and its local coordinates are\footnote{%
In this paper the Latin letters $i,$ $j,$ $k,$ ... run from $1$ to $3$. The
Einstein convention of summation is adopted all over this work.}%
\begin{equation*}
(x,y):=\left( x^{i},y^{i}\right) _{i=\overline{1,3}}=\left( \underset{\text{%
spatial coordinates}}{\underbrace{x^{1},x^{2},x^{3}}},\text{ }\underset{%
\text{spatial direction}}{\underbrace{y^{1},y^{2},y^{3}}}\right) .
\end{equation*}

Let us introduce a metric on $T\Sigma ^{3}$ inspired by optics in a
non-uniform medium (see formula (\ref{Synge-Beil})):%
\begin{equation*}
\mathfrak{g}_{ij}(x,y)=\delta _{ij}+\mathbf{\gamma }^{2}(x)y_{i}y_{j},
\end{equation*}%
where $\delta =\left( \delta _{ij}\right) =$ diag $(1,1,1)$ is the Euclidean
metric, and $y_{i}=\delta _{ir}y^{r}$. Usually, we have $\mathbf{\gamma }%
^{2}(x)=n^{2}(x)-1,$ where $n=n(x)$ is the refractive index of the
non-uniform medium (see \cite{Balan-Synge, Mir-Kaw, Neagu Optics TM}). Using
this spatial metric, below we will examine the special case of a possible
anisotropic non-relativistic dynamical model, 
which is governed by the Lagrangian (in this model one considers that the
particle has the mass $m=1$)

\begin{equation}
\begin{array}{lll}
L(x,y) & = & \dfrac{1}{2}\mathfrak{g}_{ij}(x,y)y^{i}y^{j}=\medskip \\ 
& = & \dfrac{1}{2}\left[ \delta _{ij}+\mathbf{\gamma }^{2}y_{i}y_{j}\right]
y^{i}y^{j}=\medskip \\ 
& = & \dfrac{1}{2}\delta _{ij}y^{i}y^{j}+{\dfrac{\mathbf{\gamma }^{2}}{2}}%
||y||^{4},%
\end{array}
\label{Lagragian-ARO}
\end{equation}%
where $||y||^{2}=(y^{1})^{2}+(y^{2})^{2}+(y^{3})^{2}=\delta _{ij}y^{i}y^{j}.$

\begin{remark}
Because the Euclidean metric $\delta _{ij}$ is invariant with respect to the
linear transformations of coordinates induced by the Lie group of orthogonal
transformations%
\begin{equation*}
O(3)=\left\{ A\in M_{3}(\mathbb{R})\text{ }|\text{ }A^{T}\cdot
A=I_{3}\right\} ,
\end{equation*}%
it immediately follows that the Lagrangian (\ref{Lagragian-ARO}) has a
global geometrical character with respect to these orthogonal
transformations.
\end{remark}

Following as a pattern the geometrical ideas from Lagrangian geometry of
tangent bundles \cite{Mir-An} or jet bundles \cite{Bal-Nea-Wiley}, in what
follows we construct the Riemann-Lagrange geometrical objects (the canonical
nonlinear connection, the Cartan canonical linear connection, together with
its d-torsions and d-curvatures) produced by the Lagrangian (\ref%
{Lagragian-ARO}).

\section{Geometrical objects in the non-relativistic extended dynamics}


The fundamental metrical distinguished tensor induced by the Lagrangian (\ref%
{Lagragian-ARO}) is given by%
\begin{equation*}
g_{ij}(x,y)=\frac{1}{2}\frac{\partial ^{2}L}{\partial y^{i}\partial y^{j}}%
=\sigma (x,y)\delta _{ij}+2\mathbf{\gamma }^{2}(x)y_{i}y_{j},
\end{equation*}%
where $\sigma (x,y)=(1/2)+\mathbf{\gamma }^{2}(x)||y||^{2}>0.$

\begin{remark}
Because the quadratic form $q(\xi )=\left( y_{i}y_{j}\right) \xi ^{i}\xi
^{j} $ is degenerate and has the signature $(2,1,0)$, while the Euclidean
metric $\delta (\xi )=\delta _{ij}\xi ^{i}\xi ^{j}$ is non-degenerate and
has the signature $(0,3,0)$, we easily deduce that the quadratic form $g(\xi
)=g_{ij}(x,y)\xi ^{i}\xi ^{j}$ has the constant signature $(0,3,0)$. It
follows that it is invariant under orthogonal linear transformation of
coordinates. Consequently, all the subsequent geometrical objects
constructed in this paper will have the same form in any chart of
coordinates induced by a linear transformation of coordinates produced by
the orthogonal group $O(3)$.
\end{remark}

The inverse matrix $g^{-1}=(g^{jk})$ has the entries%
\begin{equation*}
g^{jk}(x,y)=\frac{1}{\sigma (x,y)}\delta ^{jk}-\frac{2\mathbf{\gamma }^{2}(x)%
}{\sigma (x,y)\cdot \tau (x,y)}y^{j}y^{k},
\end{equation*}%
where $\delta ^{jk}=\delta _{jk}$ and $\tau (x,y)=(1/2)+3\mathbf{\gamma }%
^{2}(x)||y||^{2}=\sigma (x,y)+2\mathbf{\gamma }^{2}(x)||y||^{2}.$

\begin{proposition}
For the anisotropic 
Lagrangian (\ref{Lagragian-ARO}), the \textit{action}%
\begin{equation*}
\mathbb{E}(x(t))=\int_{a}^{b}L(x(t),y(t))dt,
\end{equation*}%
where $y=dx/dt$, produces on the tangent bundle $T\Sigma ^{3}$ the \textbf{%
canonical nonlinear connection }$N=\left( N_{j}^{i}\right) $, whose
components are%
\begin{equation}
\begin{array}{lll}
N_{j}^{i} & = & \dfrac{2\mathbf{\gamma }}{\sigma }y^{i}y_{j}(\mathbf{\gamma }%
_{s}y^{s})+\medskip \\ 
&  & +\dfrac{\mathbf{\gamma }||y||^{2}}{\sigma }\left[ \delta _{j}^{i}(%
\mathbf{\gamma }_{s}y^{s})+y^{i}\mathbf{\gamma }_{j}-\mathbf{\gamma }%
^{i}y_{j}-\dfrac{2\mathbf{\gamma }^{2}}{\sigma }y^{i}y_{j}(\mathbf{\gamma }%
_{s}y^{s})-\right. \medskip \\ 
&  & \left. -\dfrac{6\mathbf{\gamma }^{2}}{\tau }y^{i}y_{j}(\mathbf{\gamma }%
_{s}y^{s})\right] +\medskip \\ 
&  & +\dfrac{\mathbf{\gamma }^{3}||y||^{4}}{2\sigma }\left[ \dfrac{1}{\sigma 
}\mathbf{\gamma }^{i}y_{j}-\dfrac{3}{\tau }y^{i}\mathbf{\gamma }_{j}-\dfrac{3%
}{\tau }\delta _{j}^{i}(\mathbf{\gamma }_{s}y^{s})+\right. \medskip \\ 
&  & \left. +\dfrac{6\mathbf{\gamma }^{2}}{\sigma \tau ^{2}}\left( \tau
+3\sigma \right) y^{i}y_{j}(\mathbf{\gamma }_{s}y^{s})\right] ,%
\end{array}
\label{nlc-optics}
\end{equation}%
where $\mathbf{\gamma }_{s}=\partial \mathbf{\gamma }/\partial x^{s}$ and $%
\mathbf{\gamma }^{i}=\delta ^{ir}\mathbf{\gamma }_{r}=\mathbf{\gamma }_{i}$.
\end{proposition}

\begin{proof}
For the energy action functional $\mathbb{E}$, the associated Euler-Lagrange
equations can be written in the equivalent form (see \cite{Mir-An,
Bal-Nea-Wiley})%
\begin{equation}
\frac{d^{2}x^{i}}{dt^{2}}+2G^{i}\left( x^{k}(t),y^{k}(t)\right) =0,\qquad
\forall \text{ }i=\overline{1,3},  \label{Euler-Lagrange=0}
\end{equation}%
where the local components%
\begin{equation*}
\begin{array}{lll}
G^{i} & \overset{def}{=} & \dfrac{g^{ir}}{4}\left[ \dfrac{\partial ^{2}L}{%
\partial y^{r}\partial x^{s}}y^{s}-\dfrac{\partial L}{\partial x^{r}}\right]
=\medskip \\ 
& = & \dfrac{\mathbf{\gamma }}{\sigma }||y||^{2}y^{i}(\mathbf{\gamma }%
_{s}y^{s})-\dfrac{\mathbf{\gamma }}{4\sigma }||y||^{4}\mathbf{\gamma }^{i}-%
\dfrac{3\mathbf{\gamma }^{3}}{2\sigma \tau }||y||^{4}y^{i}(\mathbf{\gamma }%
_{s}y^{s})%
\end{array}%
\end{equation*}%
represent, from a geometrical point of view, a \textit{semispray} on the
tangent vector bundle $T\Sigma ^{3}$. The \textit{canonical nonlinear
connection} associated to this semispray has the components (see \cite%
{Mir-An})%
\begin{equation*}
N_{j}^{i}\overset{def}{=}\dfrac{\partial G^{i}}{\partial y^{j}}.
\end{equation*}

In conclusion, by direct computations, we find the expression (\ref%
{nlc-optics}).
\end{proof}

\begin{remark}
In an uniform medium with the constant refractive index%
\begin{equation*}
n(x)=n\in \lbrack 1,\infty ),
\end{equation*}%
we have $\mathbf{\gamma }_{s}=0$. Consequently, in this case we obtain $%
G^{i}=0$ and $N_{j}^{i}=0$.
\end{remark}

The nonlinear connection (\ref{nlc-optics}) produces the dual \textit{%
adapted bases} of d-vector fields 
\begin{equation}
\left\{ \frac{\delta }{\delta x^{i}}=\frac{\partial }{\partial x^{i}}%
-N_{i}^{r}\frac{\partial }{\partial y^{r}}\text{ };\text{ }\dfrac{\partial }{%
\partial y^{i}}\right\} \subset \mathcal{X}(T\Sigma ^{3})  \label{a-b-v}
\end{equation}%
and d-covector fields%
\begin{equation}
\left\{ dx^{i}\text{ };\text{ }\delta y^{i}=dy^{i}+N_{r}^{i}dx^{r}\right\}
\subset \mathcal{X}^{\ast }(T\Sigma ^{3}).  \label{a-b-co}
\end{equation}%
The naturalness of the geometrical adapted bases (\ref{a-b-v}) and (\ref%
{a-b-co}) is coming from the fact that, via a general transformation of
coordinates, their elements transform as\ tensors on $\Sigma ^{3}$.
Therefore, the description of all subsequent geometrical objects on the
tangent space $T\Sigma ^{3}$ (e.g., the Cartan canonical linear connection,
its torsion and curvature) will be done in local adapted components.

For instance, using the notations%
\begin{equation*}
\begin{array}{cccc}
N_{ij}:=N_{i}^{r}\delta _{rj}, & N_{i0}:=N_{ir}y^{r}, & N_{0j}:=N_{rj}y^{r},
& N_{00}:=N_{ij}y^{i}y^{j},%
\end{array}%
\end{equation*}%
by direct local computations, we obtain the following geometrical result:

\begin{proposition}
The Cartan canonical $N$-linear connection produced by the anisotropic 
Lagrangian (\ref{Lagragian-ARO}) has the adapted local components%
\begin{equation*}
C\Gamma (N)=\left( L_{jk}^{i},\text{ }C_{jk}^{i}\right) ,
\end{equation*}%
where%
\begin{equation*}
L_{jk}^{i}=-\dfrac{\mathbf{\gamma }}{\sigma }\left[ \mathbf{\gamma }\left(
\delta _{j}^{i}N_{k0}+\delta _{k}^{i}N_{j0}-\delta _{jk}\delta
^{ir}N_{r0}\right) +||y||^{2}\left( \delta _{jk}\mathbf{\gamma }^{i}-\delta
_{j}^{i}\mathbf{\gamma }_{k}-\delta _{k}^{i}\mathbf{\gamma }_{j}\right)
+\right.
\end{equation*}%
\begin{equation*}
+\mathbf{\gamma }\left\{ \left( N_{jk}+N_{kj}\right) y^{i}+\left(
N_{k}^{i}-\delta ^{ir}N_{rk}\right) y_{j}+\left( N_{j}^{i}-\delta
^{ir}N_{rj}\right) y_{k}\right\} +
\end{equation*}%
\begin{equation*}
\left. +2\left( \mathbf{\gamma }^{i}y_{j}y_{k}-y^{i}y_{j}\mathbf{\gamma }%
_{k}-y^{i}y_{k}\mathbf{\gamma }_{j}\right) \right] +\dfrac{2\mathbf{\gamma }%
^{3}y^{i}}{\sigma \tau }\left[ \mathbf{\gamma }\left(
y_{j}N_{k0}+y_{k}N_{j0}-\delta _{jk}N_{00}\right) +\right.
\end{equation*}%
\begin{equation*}
+||y||^{2}\left( \delta _{jk}\mathbf{\gamma }_{r}y^{r}-y_{j}\mathbf{\gamma }%
_{k}-y_{k}\mathbf{\gamma }_{j}\right) +2\left( y_{j}y_{k}\mathbf{\gamma }%
_{r}y^{r}-y_{j}\mathbf{\gamma }_{k}||y||^{2}-y_{k}\mathbf{\gamma }%
_{j}||y||^{2}\right) +
\end{equation*}%
\begin{equation*}
\left. +\mathbf{\gamma }\left\{ \left( N_{jk}+N_{kj}\right) ||y||^{2}+\left(
N_{k0}-N_{0k}\right) y_{j}+\left( N_{j0}-N_{0j}\right) y_{k}\right\} \right]
,
\end{equation*}%
\medskip 
\begin{equation*}
C_{jk}^{i}=\frac{\mathbf{\gamma }^{2}}{\sigma }\left( \delta
_{j}^{i}y_{k}+\delta _{k}^{i}y_{j}+\delta _{jk}y^{i}\right) -\frac{2\mathbf{%
\gamma }^{4}}{\sigma \tau }\left( ||y||^{2}\delta _{jk}+2y_{j}y_{k}\right)
y^{i}.
\end{equation*}
\end{proposition}

\begin{proof}
The adapted components of the Cartan canonical connection are given by the
general formulas (see \cite{Mir-An})%
\begin{equation*}
L_{jk}^{i}\overset{def}{=}\frac{g^{ir}}{2}\left( \frac{\delta g_{jr}}{\delta
x^{k}}+\frac{\delta g_{kr}}{\delta x^{j}}-\frac{\delta g_{jk}}{\delta x^{r}}%
\right) ,
\end{equation*}%
\begin{equation*}
C_{jk}^{i}\overset{def}{=}\frac{g^{ir}}{2}\left( \frac{\partial g_{jr}}{%
\partial y^{k}}+\frac{\partial g_{kr}}{\partial y^{j}}-\frac{\partial g_{jk}%
}{\partial y^{r}}\right) =\frac{g^{ir}}{2}\frac{\partial g_{jr}}{\partial
y^{k}}.
\end{equation*}%
Using the derivative operators (\ref{a-b-v}), the direct calculations lead
us to the required results.
\end{proof}

The Cartan canonical $N$-linear connection produced by the anisotropic 
Lagrangian (\ref{Lagragian-ARO}) is characterized by \textit{three}\textbf{\ 
}effective local torsion d-tensors, namely%
\begin{equation*}
R_{jk}^{i}\overset{def}{=}{\dfrac{\delta N_{j}^{i}}{\delta x^{k}}}-{\dfrac{%
\delta N_{k}^{i}}{\delta x^{j}}},\quad P_{jk}^{i}{\overset{def}{=}{\dfrac{%
\partial N_{j}^{i}}{\partial y^{k}}}-L_{kj}^{i}},\quad C_{jk}^{i},
\end{equation*}%
and \textit{three} effective local curvature d-tensors:%
\begin{equation*}
{R_{jkl}^{i}\overset{def}{=}{\dfrac{\delta L{_{jk}^{i}}}{\delta x^{l}}}-{%
\dfrac{\delta L{_{jl}^{i}}}{\delta x^{k}}}%
+L_{jk}^{r}L_{rl}^{i}-L_{jl}^{r}L_{rk}^{i}+C{_{jr}^{i}}}R_{kl}^{r},
\end{equation*}%
\begin{equation*}
{P_{jkl}^{i}\overset{def}{=}{\dfrac{\partial L{_{jk}^{i}}}{\partial y^{l}}}-C%
{_{jl|k}^{i}}+{C{_{jr}^{i}}}}P_{kl}^{r}{,}
\end{equation*}%
\begin{equation*}
S{_{jkl}^{i}\overset{def}{=}{\dfrac{\partial C{_{jk}^{i}}}{\partial y^{l}}}-{%
\dfrac{\partial C{_{jl}^{i}}}{\partial y^{k}}}%
+C_{jk}^{r}C_{rl}^{i}-C_{jl}^{r}C_{rk}^{i},}
\end{equation*}%
where%
\begin{equation*}
{C{_{jl|k}^{i}}\overset{def}{=}}\frac{\delta {C{_{jl}^{i}}}}{\delta x^{k}}+{C%
{_{jl}^{r}}}L{_{rk}^{i}}-{C{_{rl}^{i}}}L{_{jk}^{r}}-{C{_{jr}^{i}}}L{%
_{lk}^{r}.}
\end{equation*}

At the end of this Section, we would like to point out that in our
anisotropic Lagrangian geometrical theory for non-relativistic extended
dynamics, the classical Riemannian Levi-Civita connection (which in our case
would be attached even to the Euclidean metric $\delta _{ij}$) is replaced
by Cartan canonical connection (associated to the perturbed Lagrangian (\ref%
{Lagragian-ARO})). It is well known that in Finsler-Lagrange geometrical
framework (see \cite{Mir-An}) there are a lot of important linear
distinguished (d-) connections (Cartan, Berwald, Chern-Rund or Hashiguchi,
for instance). However, the use of the Cartan canonical connection is
preferred because it is the single linear d-connection which is a \textit{%
metrical connection}, like the Levi-Civita connection in the classical
Riemannian framework. An important difference between these two connections
is that the Cartan connection is only partial torsion-free because the
Poisson brackets of the distinguished vector fields $\delta /\delta x^{i}$
are not generally equal to zero.

\section{The equations of motion in the anisotropic non-relativistic dynamics%
}

The Euler-Lagrange equations (\ref{Euler-Lagrange=0}), for%
\begin{equation*}
\left( y^{1},y^{2},y^{3}\right) =\left( V^{1},V^{2},V^{3}\right) :=V,\quad
V^{i}=\frac{dx^{i}}{dt},
\end{equation*}%
\begin{equation*}
v^{2}=|V|^{2}=\left( V^{1}\right) ^{2}+\left( V^{2}\right) ^{2}+\left(
V^{3}\right) ^{2},
\end{equation*}%
lead us to the following \textit{anisotropic equations of motion} for a
particle:%
\begin{equation}
\dfrac{dV^{i}}{dt}+\frac{4\mathbf{\gamma }v^{2}\left( 1+3\mathbf{\gamma }%
^{2}v^{2}\right) }{\left( 1+2\mathbf{\gamma }^{2}v^{2}\right) \left( 1+6%
\mathbf{\gamma }^{2}v^{2}\right) }\left( \mathbf{\gamma }_{s}V^{s}V^{i}%
\right) -\frac{\mathbf{\gamma }v^{4}}{1+2\mathbf{\gamma }^{2}v^{2}}\mathbf{%
\gamma }^{i}=0,  \label{eq-of-motin}
\end{equation}%
where $i\in \{1,2,3\}.$

\begin{remark}
\textbf{In a uniform medium with the constant refractive index }$%
n(x)=n_{0}\in \lbrack 1,\infty )$, the above anisotropic equations of motion
simplify as%
\begin{equation*}
\dfrac{dV^{i}}{dt}=0\Leftrightarrow V=\left( V^{1},V^{2},V^{3}\right) =\text{%
\emph{constant}}\Leftrightarrow
\end{equation*}%
\begin{equation*}
\dfrac{dx^{i}}{dt}=V^{i}\Leftrightarrow x(t)=\left(
V^{1}t+x_{0}^{1},V^{2}t+x_{0}^{2},V^{3}t+x_{0}^{3}\right) ,
\end{equation*}%
where $\left( x_{0}^{1},x_{0}^{2},x_{0}^{3}\right) $ = \emph{constant}. It
follows that, in this case, \textbf{the particles are moving only on
straight lines}.
\end{remark}

\subsection{Particular solutions with cylindrical symmetry}

Let us investigate now the case of a non-uniform medium with cylindrical
symmetry. This means that we have $n^{2}=1+f^{2}(\rho )\Leftrightarrow 
\mathbf{\gamma }=f(\rho ),$ where $f:(0,\infty )\rightarrow (0,\infty )$ is
an arbitrary non-constant smooth function, and%
\begin{equation*}
\rho ^{2}=\left( x^{1}\right) ^{2}+\left( x^{2}\right) ^{2}>0.
\end{equation*}

In such a special context, let us search for solutions of (\ref{eq-of-motin}%
) in cylindrical coordinates:%
\begin{equation*}
\begin{array}{ccc}
x^{1}(t)=\rho (t)\cos \phi (t), & x^{2}(t)=\rho (t)\sin \phi (t), & 
x^{3}(t)=\zeta (t),%
\end{array}%
\end{equation*}%
where $\phi \in \lbrack 0,2\pi ]$ and $\zeta \in \mathbb{R}$. By direct
computations, we deduce that the equations (\ref{eq-of-motin}) rewrite as%
\begin{equation}
\begin{array}{c}
\cos \phi \left( \ddot{\rho}\right) -2\sin \phi \left( \dot{\rho}\dot{\phi}%
\right) -\rho \cos \phi \left( \dot{\phi}\right) ^{2}-\rho \sin \phi \left( 
\ddot{\phi}\right) +\medskip \\ 
+\dfrac{4ff^{\prime }v^{2}\left( 1+3f^{2}v^{2}\right) }{\left(
1+2f^{2}v^{2}\right) \left( 1+6f^{2}v^{2}\right) }\left( \cos \phi \left( 
\dot{\rho}\right) ^{2}-\rho \sin \phi \left( \dot{\rho}\dot{\phi}\right)
\right) -\medskip \\ 
-\dfrac{ff^{\prime }v^{4}}{1+2f^{2}v^{2}}\cos \phi =0,%
\end{array}
\label{motion-1}
\end{equation}%
\begin{equation}
\begin{array}{c}
\sin \phi \left( \ddot{\rho}\right) +2\cos \phi \left( \dot{\rho}\dot{\phi}%
\right) -\rho \sin \phi \left( \dot{\phi}\right) ^{2}+\rho \cos \phi \left( 
\ddot{\phi}\right) +\medskip \\ 
+\dfrac{4ff^{\prime }v^{2}\left( 1+3f^{2}v^{2}\right) }{\left(
1+2f^{2}v^{2}\right) \left( 1+6f^{2}v^{2}\right) }\left( \sin \phi \left( 
\dot{\rho}\right) ^{2}+\rho \cos \phi \left( \dot{\rho}\dot{\phi}\right)
\right) -\medskip \\ 
-\dfrac{ff^{\prime }v^{4}}{1+2f^{2}v^{2}}\sin \phi =0,%
\end{array}
\label{motion-2}
\end{equation}%
\begin{equation}
\ddot{\zeta}+\dfrac{4ff^{\prime }v^{2}\left( 1+3f^{2}v^{2}\right) }{\left(
1+2f^{2}v^{2}\right) \left( 1+6f^{2}v^{2}\right) }\left( \dot{\rho}\dot{\zeta%
}\right) =0,  \label{motion-3}
\end{equation}%
where%
\begin{equation*}
v^{2}=\left( \dot{\rho}\right) ^{2}+\rho ^{2}\left( \dot{\phi}\right)
^{2}+\left( \dot{\zeta}\right) ^{2}.
\end{equation*}

In what follows, we will look for some particular solutions of the system of
differential equations (\ref{motion-1} - \ref{motion-3}).

\textbf{Case 1: }Let us consider that $\dot{\rho}=0$ and $v\neq 0$. We will
look for solutions of the form%
\begin{equation*}
c(\zeta )=\left( \rho =\text{constant},\text{ }\phi =\phi (\zeta ),\text{ }%
\zeta \in \mathbb{R}\right) \Rightarrow \dot{c}(\zeta )=\left( \dot{\rho}=0,%
\text{ }\dot{\phi}=\frac{d\phi }{d\zeta },\text{ }\dot{\zeta}=1\right) .
\end{equation*}%
These conditions imply%
\begin{equation*}
v^{2}=\rho ^{2}\left( \frac{d\phi }{d\zeta }\right) ^{2}+1.
\end{equation*}

In this case the anisotropic equations of motion reduce to%
\begin{equation}
\left\{ 
\begin{array}{l}
-\rho \cos \phi \left( \dfrac{d\phi }{d\zeta }\right) ^{2}-\rho \sin \phi
\left( \dfrac{d^{2}\phi }{d\zeta ^{2}}\right) -\dfrac{ff^{\prime }v^{4}}{%
1+2f^{2}v^{2}}\cos \phi =0\medskip \\ 
-\rho \sin \phi \left( \dfrac{d\phi }{d\zeta }\right) ^{2}+\rho \cos \phi
\left( \dfrac{d^{2}\phi }{d\zeta ^{2}}\right) -\dfrac{ff^{\prime }v^{4}}{%
1+2f^{2}v^{2}}\sin \phi =0.%
\end{array}%
\right.  \label{steluta-0}
\end{equation}%
Multiplying the first equation of (\ref{steluta-0}) by $\left( -\sin \phi
\right) $, and the second by $\left( \cos \phi \right) $, by summing we get
the equation:%
\begin{equation}
\rho \cdot \left( \dfrac{d^{2}\phi }{d\zeta ^{2}}\right) =0\Leftrightarrow 
\dfrac{d^{2}\phi }{d\zeta ^{2}}=0\Leftrightarrow \phi (\zeta )=\omega \zeta
+\phi _{0},  \label{S1}
\end{equation}%
where $\omega ,$ $\phi _{0}\in \mathbb{R}$ are arbitrary constants (here $%
\omega $ has the physical meaning of angular velocity). Consequently, the
function (\ref{S1}) is a solution for the equations of motion (\ref%
{steluta-0}) if and only if%
\begin{equation*}
\left\{ 
\begin{array}{l}
\left( \rho \cos \phi \right) \left( \omega ^{2}\right) +\dfrac{ff^{\prime
}v^{4}}{1+2f^{2}v^{2}}\cos \phi =0\medskip \\ 
\left( \rho \sin \phi \right) \left( \omega ^{2}\right) +\dfrac{ff^{\prime
}v^{4}}{1+2f^{2}v^{2}}\sin \phi =0%
\end{array}%
\right. \Leftrightarrow
\end{equation*}%
\begin{equation}
\Leftrightarrow \rho \omega ^{2}+\dfrac{ff^{\prime }v^{4}}{1+2f^{2}v^{2}}=0,
\label{eqc-1}
\end{equation}%
where $v^{2}=\rho ^{2}\omega ^{2}+1.$

\begin{enumerate}
\item If $2f^{2}-1<0$ and 
\begin{equation*}
2f+\rho f^{\prime }\in \left[ -\frac{\left( 2f^{2}-1\right) ^{2}}{4f}%
,0\right) \cup \left( 0,2f\right) ,
\end{equation*}%
then the solutions of the equation (\ref{eqc-1}) are%
\begin{equation*}
\omega _{0}=\pm \frac{1}{\rho }\sqrt{\frac{2f^{2}-1\pm \sqrt{\Delta }}{%
2f\left( 2f+\rho f^{\prime }\right) }-1},
\end{equation*}%
where $\Delta =\left( 1-2f^{2}\right) ^{2}+4f\left( 2f+\rho f^{\prime
}\right) .$

\item If $2f^{2}-1>0$ and $2f+\rho f^{\prime }\in \left( 0,2f\right) ,$ then
the solutions of (\ref{eqc-1}) are%
\begin{equation*}
\omega _{0}=\pm \frac{1}{\rho }\sqrt{\frac{2f^{2}-1+\sqrt{\Delta }}{2f\left(
2f+\rho f^{\prime }\right) }-1}.
\end{equation*}
\end{enumerate}

In conclusion, the solutions of equations of motion (\ref{steluta-0}) are%
\begin{equation*}
c(t)=\left( \rho =\rho _{0}=\text{constant},\text{ }\phi =\omega _{0}t+\phi
_{0},\text{ }\zeta =t\right) \Rightarrow
\end{equation*}%
\begin{equation*}
\Rightarrow 
\begin{array}{ccc}
x^{1}(t)=\rho _{0}\cos \phi (t), & x^{2}(t)=\rho _{0}\sin \phi (t), & 
x^{3}(t)=t.%
\end{array}%
\end{equation*}%
It follows that in this case \textit{the particles move on some \textbf{%
circular helices}}. Note that these circular helices are some very specific
trajectories (non-evident in advance) because the values of the angular
velocity $\omega _{0}$ are completely determined by the initial function $%
f(\rho )$. Solutions corresponding to the sign plus and minus, which are
intimately related to the orientation of rotation (right-handed and
left-handed respectively).

\textbf{Case 2: }Let it be $\dot{\phi}=0$ and $v\neq 0$. We will look for
solutions of the form%
\begin{equation*}
c(\zeta )=\left( \rho =\rho (\zeta ),\text{ }\phi =\text{constant},\text{ }%
\zeta \in \mathbb{R}\right) \Rightarrow \dot{c}(\zeta )=\left( \dot{\rho}=%
\frac{d\rho }{d\zeta },\text{ }\dot{\phi}=0,\text{ }\dot{\zeta}=1\right) .
\end{equation*}%
These conditions imply%
\begin{equation*}
v^{2}=\left( \frac{d\rho }{d\zeta }\right) ^{2}+1.
\end{equation*}

In this case the anisotropic equations of motion reduce to%
\begin{equation}
\left\{ 
\begin{array}{l}
\begin{array}{c}
\cos \phi \left( \dfrac{d^{2}\rho }{d\zeta ^{2}}\right) +\dfrac{4ff^{\prime
}v^{2}\left( 1+3f^{2}v^{2}\right) }{\left( 1+2f^{2}v^{2}\right) \left(
1+6f^{2}v^{2}\right) }\cos \phi \left( \dfrac{d\rho }{d\zeta }\right)
^{2}-\medskip \\ 
-\dfrac{ff^{\prime }v^{4}}{1+2f^{2}v^{2}}\cos \phi =0%
\end{array}%
\medskip \\ 
\begin{array}{c}
\sin \phi \left( \dfrac{d^{2}\rho }{d\zeta ^{2}}\right) +\dfrac{4ff^{\prime
}v^{2}\left( 1+3f^{2}v^{2}\right) }{\left( 1+2f^{2}v^{2}\right) \left(
1+6f^{2}v^{2}\right) }\sin \phi \left( \dfrac{d\rho }{d\zeta }\right)
^{2}-\medskip \\ 
-\dfrac{ff^{\prime }v^{4}}{1+2f^{2}v^{2}}\sin \phi =0%
\end{array}
\\ 
\left( \dfrac{d\rho }{d\zeta }\right) f^{\prime }=0.%
\end{array}%
\right.  \label{steluta}
\end{equation}%
The last equation of (\ref{steluta}) implies $d\rho /d\zeta =0$ or $%
f^{\prime }=0.$ In both situations the solutions of the equations of motion
are%
\begin{equation*}
c(t)=\left( \rho =\rho _{0}=\text{constant},\text{ }\phi =\phi _{0}\in
\lbrack 0,2\pi ],\text{ }\zeta =t\in \mathbb{R}\right) \Rightarrow
\end{equation*}%
\begin{equation*}
\Rightarrow 
\begin{array}{ccc}
x^{1}(t)=\rho _{0}\cos \phi _{0}, & x^{2}(t)=\rho _{0}\sin \phi _{0}, & 
x^{3}(t)=t.%
\end{array}%
\end{equation*}%
Consequently, in this case \textit{the trajectories are \textbf{the
generators of the right circular cylinders}}%
\begin{equation*}
\left( x^{1}\right) ^{2}+\left( x^{2}\right) ^{2}=\rho _{0}^{2},
\end{equation*}%
\textit{where }$\rho _{0}>0$\textit{\ is a solution of the equation }$%
f^{\prime }(\rho )=0$.

\textbf{Case 3: }Let us consider that $\dot{\zeta}=0$ and $v\neq 0$. We will
look for solutions of the form%
\begin{equation*}
c(\phi )=\left( \rho =\rho (\phi ),\text{ }\phi :=t\in \lbrack 0,2\pi ],%
\text{ }\zeta =\text{constant}\right) \Rightarrow
\end{equation*}%
\begin{equation*}
\Rightarrow \dot{c}(\phi )=\left( \dot{\rho}=\frac{d\rho }{d\phi },\text{ }%
\dot{\phi}=1,\text{ }\dot{\zeta}=0\right) .
\end{equation*}%
These conditions imply%
\begin{equation*}
v^{2}=\left( \frac{d\rho }{dt}\right) ^{2}+\rho ^{2}.
\end{equation*}

In this case the anisotropic equations of motion reduce to%
\begin{equation}
\left\{ 
\begin{array}{l}
\begin{array}{c}
\cos t\left( \dfrac{d^{2}\rho }{dt^{2}}\right) -2\sin t\left( \dfrac{d\rho }{%
dt}\right) -\rho \cos t-\dfrac{ff^{\prime }v^{4}}{1+2f^{2}v^{2}}\cos
t+\medskip \\ 
\dfrac{4ff^{\prime }v^{2}\left( 1+3f^{2}v^{2}\right) }{\left(
1+2f^{2}v^{2}\right) \left( 1+6f^{2}v^{2}\right) }\left( \cos t\left( \dfrac{%
d\rho }{dt}\right) ^{2}-\rho \sin t\left( \dfrac{d\rho }{dt}\right) \right)
=0%
\end{array}%
\medskip \\ 
\begin{array}{c}
\sin t\left( \dfrac{d^{2}\rho }{dt^{2}}\right) +2\cos t\left( \dfrac{d\rho }{%
dt}\right) -\rho \sin t-\dfrac{ff^{\prime }v^{4}}{1+2f^{2}v^{2}}\sin
t+\medskip \\ 
\dfrac{4ff^{\prime }v^{2}\left( 1+3f^{2}v^{2}\right) }{\left(
1+2f^{2}v^{2}\right) \left( 1+6f^{2}v^{2}\right) }\left( \sin t\left( \dfrac{%
d\rho }{dt}\right) ^{2}+\rho \cos t\left( \dfrac{d\rho }{dt}\right) \right)
=0.%
\end{array}%
\end{array}%
\right.  \label{steluta-1}
\end{equation}%
Multiplying the first equation of (\ref{steluta-1}) by $\left( -\sin \phi
\right) $, and the second by $\left( \cos \phi \right) $, by summing we get
the equation:%
\begin{equation*}
\left( \dfrac{d\rho }{dt}\right) \left( 1+\dfrac{2\rho ff^{\prime
}v^{2}\left( 1+3f^{2}v^{2}\right) }{\left( 1+2f^{2}v^{2}\right) \left(
1+6f^{2}v^{2}\right) }\right) =0.
\end{equation*}

\begin{enumerate}
\item If $d\rho /dt=0$, then the equations of motion (\ref{steluta-1}) become%
\begin{equation*}
\left\{ 
\begin{array}{l}
\rho \cos t+\dfrac{ff^{\prime }v^{4}}{1+2f^{2}v^{2}}\cos t=0\medskip \\ 
\rho \sin t+\dfrac{ff^{\prime }v^{4}}{1+2f^{2}v^{2}}\sin t=0.%
\end{array}%
\right. \Leftrightarrow \rho +\dfrac{ff^{\prime }v^{4}}{1+2f^{2}v^{2}}%
=0\Leftrightarrow
\end{equation*}%
\begin{equation}
\Leftrightarrow 1+2f^{2}\rho ^{2}+ff^{\prime }\rho ^{3}=0,  \label{eq-2}
\end{equation}%
where $v^{2}=\rho ^{2}$. The corresponding solutions are%
\begin{equation*}
c(t)=\left( \rho =\rho _{0}>0,\text{ }\phi =t\in \lbrack 0,2\pi ],\text{ }%
\zeta =\zeta _{0}\in \mathbb{R}\right) \Rightarrow
\end{equation*}%
\begin{equation*}
\Rightarrow 
\begin{array}{ccc}
x^{1}(t)=\rho _{0}\cos t, & x^{2}(t)=\rho _{0}\sin t, & x^{3}(t)=\zeta _{0}.%
\end{array}%
\end{equation*}%
It follows that in this case \textit{the trajectories are \textbf{circles
situated in planes parallel with the plane }}$xOy$\textbf{, }\textit{\textbf{%
having the centers situated on the axis }}$Oz$\textbf{, }\textit{\textbf{and
the radii }}$\rho _{0}>0$\textbf{, }\textit{\textbf{where }}$\rho _{0}$%
\textbf{\ }\textit{\textbf{is a solution of (\ref{eq-2}).}}

\item If $d\rho /dt\neq 0$, then we get 
\begin{equation}
1+\dfrac{2\rho ff^{\prime }v^{2}\left( 1+3f^{2}v^{2}\right) }{\left(
1+2f^{2}v^{2}\right) \left( 1+6f^{2}v^{2}\right) }=0.  \label{eqc-2}
\end{equation}%
The preceding equation has solution if and only if%
\begin{equation*}
2f+\rho f^{\prime }\in \left( -\frac{1+4\rho ^{2}f^{2}}{2\rho ^{2}f\left(
1+3\rho ^{2}f^{2}\right) },0\right) .
\end{equation*}%
In this case, the solutions of the equation (\ref{eqc-2}) are%
\begin{equation}
\frac{d\rho }{dt}=\pm \sqrt{\frac{-4f^{2}-\rho ff^{\prime }-\sqrt{\Delta
^{\prime }}}{6f^{3}\left( 2f+\rho f^{\prime }\right) }-\rho ^{2}}:=F\left(
\rho \right) ,  \label{eq-3}
\end{equation}%
where $\Delta ^{\prime }=4f^{4}+\rho ^{2}f^{2}\left( f^{\prime }\right)
^{2}+2\rho f^{3}f^{\prime }$. Imposing as a solution of (\ref{eq-3}) to be
also a solution for the system (\ref{steluta-1}), we get%
\begin{equation*}
FF^{\prime }-\rho -\frac{2F^{2}}{\rho }=\frac{ff^{\prime }\left( F^{2}+\rho
^{2}\right) ^{2}}{1+2f^{2}\left( F^{2}+\rho ^{2}\right) }\Rightarrow
\end{equation*}%
\begin{equation*}
\Rightarrow \rho =\text{constant (contradiction!)}.
\end{equation*}%
So, the case $d\rho /dt\neq 0$ does not provide us any solution.
\end{enumerate}

\subsection{Particular solutions with spherical symmetry}

We investigate now the case of a non-uniform medium with spherical symmetry.
This means that we have $n^{2}=1+f^{2}(r)\Leftrightarrow \mathbf{\gamma }%
=f(r),$ where $f:(0,\infty )\rightarrow (0,\infty )$ is an arbitrary
non-constant smooth function, and%
\begin{equation*}
r^{2}=|x|^{2}=\left( x^{1}\right) ^{2}+\left( x^{2}\right) ^{2}+\left(
x^{3}\right) ^{2}>0.
\end{equation*}

\begin{example}
In the study of mirages from heat above warm surfaces, one works with
indices of refraction which have spherical symmetry (see \cite{Minnaer, Stam}%
):%
\begin{equation*}
n(x)=1+\varepsilon e^{-r^{2}/\sigma ^{2}}\Leftrightarrow f(r)=\sqrt{%
2\varepsilon e^{-r^{2}/\sigma ^{2}}+\varepsilon ^{2}e^{-2r^{2}/\sigma ^{2}}},
\end{equation*}%
where $\varepsilon >0$ and $\sigma >0$ are some given constants.
\end{example}

In such a special context, let us search for solutions of (\ref{eq-of-motin}%
) in spherical coordinates:%
\begin{equation*}
x^{1}(t)=r(t)\sin \theta (t)\cos \phi (t),\text{ \ }x^{2}(t)=r(t)\sin \theta
(t)\sin \phi (t),\text{ \ }x^{3}(t)=r(t)\cos \theta (t),
\end{equation*}%
where $\theta \in \lbrack 0,\pi ]$ and $\phi \in \lbrack 0,2\pi ]$. By
direct computations, we deduce that the equations (\ref{eq-of-motin})
rewrite as%
\begin{equation}
\begin{array}{c}
\sin \theta \cos \phi \left( \ddot{r}\right) +r\cos \theta \cos \phi \left( 
\ddot{\theta}\right) -r\sin \theta \sin \phi \left( \ddot{\phi}\right)
+\medskip \\ 
+2\cos \theta \cos \phi \left( \dot{r}\dot{\theta}\right) -2\sin \theta \sin
\phi \left( \dot{r}\dot{\phi}\right) -2r\cos \theta \sin \phi \left( \dot{%
\theta}\dot{\phi}\right) -\medskip \\ 
-r\sin \theta \cos \phi \left( \left( \dot{\theta}\right) ^{2}+\left( \dot{%
\phi}\right) ^{2}\right) +\dfrac{4ff^{\prime }v^{2}\left(
1+3f^{2}v^{2}\right) }{\left( 1+2f^{2}v^{2}\right) \left(
1+6f^{2}v^{2}\right) }\cdot \medskip \\ 
\cdot \left( \sin \theta \cos \phi \left( \dot{r}\right) ^{2}+r\cos \theta
\cos \phi \left( \dot{r}\dot{\theta}\right) -r\sin \theta \sin \phi \left( 
\dot{r}\dot{\phi}\right) \right) -\medskip \\ 
-\dfrac{ff^{\prime }v^{4}}{1+2f^{2}v^{2}}\sin \theta \cos \phi =0,%
\end{array}
\label{motion-1-cyl}
\end{equation}%
\begin{equation}
\begin{array}{c}
\sin \theta \sin \phi \left( \ddot{r}\right) +r\cos \theta \sin \phi \left( 
\ddot{\theta}\right) +r\sin \theta \cos \phi \left( \ddot{\phi}\right)
+\medskip \\ 
+2\cos \theta \sin \phi \left( \dot{r}\dot{\theta}\right) +2\sin \theta \cos
\phi \left( \dot{r}\dot{\phi}\right) +2r\cos \theta \cos \phi \left( \dot{%
\theta}\dot{\phi}\right) -\medskip \\ 
-r\sin \theta \sin \phi \left( \left( \dot{\theta}\right) ^{2}+\left( \dot{%
\phi}\right) ^{2}\right) +\dfrac{4ff^{\prime }v^{2}\left(
1+3f^{2}v^{2}\right) }{\left( 1+2f^{2}v^{2}\right) \left(
1+6f^{2}v^{2}\right) }\cdot \medskip \\ 
\cdot \left( \sin \theta \sin \phi \left( \dot{r}\right) ^{2}+r\cos \theta
\sin \phi \left( \dot{r}\dot{\theta}\right) +r\sin \theta \cos \phi \left( 
\dot{r}\dot{\phi}\right) \right) -\medskip \\ 
-\dfrac{ff^{\prime }v^{4}}{1+2f^{2}v^{2}}\sin \theta \sin \phi =0,%
\end{array}
\label{motion-2-cyl}
\end{equation}%
\begin{equation}
\begin{array}{c}
\cos \theta \left( \ddot{r}\right) -r\sin \theta \left( \ddot{\theta}\right)
-2\sin \theta \left( \dot{r}\dot{\theta}\right) -r\cos \theta \left( \dot{%
\theta}\right) ^{2}+\medskip \\ 
+\dfrac{4ff^{\prime }v^{2}\left( 1+3f^{2}v^{2}\right) }{\left(
1+2f^{2}v^{2}\right) \left( 1+6f^{2}v^{2}\right) }\left( \cos \theta \left( 
\dot{r}\right) ^{2}-r\sin \theta \left( \dot{r}\dot{\theta}\right) \right)
-\medskip \\ 
-\dfrac{ff^{\prime }v^{4}}{1+2f^{2}v^{2}}\cos \theta =0,%
\end{array}
\label{motion-3-cyl}
\end{equation}%
where%
\begin{equation*}
v^{2}=\left( \dot{r}\right) ^{2}+r^{2}\left( \dot{\theta}\right)
^{2}+r^{2}\sin ^{2}\theta \left( \dot{\phi}\right) ^{2}.
\end{equation*}

Let us look for some particular solutions of the DEs system (\ref%
{motion-1-cyl} - \ref{motion-3-cyl}).

\textbf{Case 1: }Let it be $\dot{r}=0$ and $v\neq 0$. We will look for
solutions of the form%
\begin{equation*}
c(\phi )=\left( r=\text{constant},\text{ }\theta =\theta (\phi ),\text{ }%
\phi :=t\in \lbrack 0,2\pi ]\right) \Rightarrow
\end{equation*}%
\begin{equation*}
\Rightarrow \dot{c}(\phi )=\left( \dot{r}=0,\text{ }\dot{\theta}=\frac{%
d\theta }{d\phi },\text{ }\dot{\phi}=1\right) .
\end{equation*}%
These conditions imply%
\begin{equation*}
v^{2}=r^{2}\left( \left( \frac{d\theta }{dt}\right) ^{2}+\sin ^{2}\theta
\right) .
\end{equation*}

In this case the anisotropic equations of motion reduce to%
\begin{equation}
\left\{ 
\begin{array}{l}
\begin{array}{l}
r\cos \theta \cos t\left( \dfrac{d^{2}\theta }{dt^{2}}\right) -2r\cos \theta
\sin t\left( \dfrac{d\theta }{dt}\right) -\medskip \\ 
-r\sin \theta \cos t\left( \left( \dfrac{d\theta }{dt}\right) ^{2}+1\right) -%
\dfrac{ff^{\prime }v^{4}}{1+2f^{2}v^{2}}\sin \theta \cos t=0%
\end{array}%
\medskip \\ 
\begin{array}{l}
r\cos \theta \sin t\left( \dfrac{d^{2}\theta }{dt^{2}}\right) +2r\cos \theta
\cos t\left( \dfrac{d\theta }{dt}\right) -\medskip \\ 
-r\sin \theta \sin t\left( \left( \dfrac{d\theta }{dt}\right) ^{2}+1\right) -%
\dfrac{ff^{\prime }v^{4}}{1+2f^{2}v^{2}}\sin \theta \sin t=0%
\end{array}%
\medskip \\ 
r\sin \theta \left( \dfrac{d^{2}\theta }{dt^{2}}\right) +r\cos \theta \left( 
\dfrac{d\theta }{dt}\right) ^{2}+\dfrac{ff^{\prime }v^{4}}{1+2f^{2}v^{2}}%
\cos \theta =0.%
\end{array}%
\right.  \label{em-1}
\end{equation}%
Multiplying the first equation of (\ref{em-1}) by $\left( -\sin t\right) $,
and the second by $\left( \cos t\right) $, by summing we get the equation:%
\begin{equation*}
2r\cos \theta \left( \dfrac{d\theta }{dt}\right) =0.
\end{equation*}

\begin{enumerate}
\item If $\cos \theta =0\Leftrightarrow \theta =\pi /2$, then the equations
of motion (\ref{em-1}) become%
\begin{equation*}
\left\{ 
\begin{array}{l}
r\sin \theta \cos t+\dfrac{ff^{\prime }v^{4}}{1+2f^{2}v^{2}}\sin \theta \cos
t=0\medskip \\ 
r\sin \theta \sin t+\dfrac{ff^{\prime }v^{4}}{1+2f^{2}v^{2}}\sin \theta \sin
t=0,%
\end{array}%
\right. \Leftrightarrow
\end{equation*}%
\begin{equation}
\Leftrightarrow 1+\dfrac{ff^{\prime }r^{3}}{1+2f^{2}r^{2}}=0,  \label{eqs-1}
\end{equation}%
where $v^{2}=r^{2}$. It follows that in this case the solutions of the
anisotropic equations of motion are%
\begin{equation*}
c(t)=\left( r=r_{0}=\text{constant},\text{ }\theta =\frac{\pi }{2},\text{ }%
\phi =t\right) \Rightarrow
\end{equation*}%
\begin{equation*}
\Rightarrow 
\begin{array}{ccc}
x^{1}(t)=r_{0}\cos t, & x^{2}(t)=r_{0}\sin t, & x^{3}(t)=0,%
\end{array}%
\end{equation*}%
where $r_{0}>0$\ is a solution of the equation (\ref{eqs-1}). It follows
that in this case \textit{the particles move on the \textbf{circles situated
in the plane }}$xOy$\textbf{, }\textit{\textbf{which have the centers in
origin and the radii equal to the roots of the equation (\ref{eqs-1})}}%
\textbf{. }We recall that in the case of cylinder solutions we obtained some 
\textit{circles situated in planes parallel with the plane }$xOy$\textbf{,}
having the centers situated on the axis $Oz$, and the radii $\rho _{0}>0$,
where $\rho _{0}$\ is a solution of the same equation.

\item If we have $d\theta /dt=0$, then the anisotropic equations of motion
become%
\begin{equation*}
\left\{ 
\begin{array}{l}
r\sin \theta \cos t+\dfrac{ff^{\prime }v^{4}}{1+2f^{2}v^{2}}\sin \theta \cos
t=0\medskip \\ 
r\sin \theta \sin t+\dfrac{ff^{\prime }v^{4}}{1+2f^{2}v^{2}}\sin \theta \sin
t=0\medskip \\ 
\dfrac{ff^{\prime }v^{4}}{1+2f^{2}v^{2}}\cos \theta =0.%
\end{array}%
\right.
\end{equation*}%
This system implies $\theta =\pi /2$ (we recover then the above Case 1.) or $%
f^{\prime }=0$. So, the new corresponding solutions are%
\begin{equation*}
c(t)=\left( r=r_{0}>0,\text{ }\theta =\theta _{0}\in \lbrack 0,\pi
]\backslash \{\pi /2\},\text{ }\phi =t\right) \Rightarrow
\end{equation*}%
\begin{equation*}
\Rightarrow 
\begin{array}{ccc}
x^{1}(t)=\left( r_{0}\sin \theta _{0}\right) \cos t, & x^{2}(t)=\left(
r_{0}\sin \theta _{0}\right) \sin t, & x^{3}(t)=r_{0}\cos \theta _{0}.%
\end{array}%
\end{equation*}%
\textit{These trajectories are \textbf{circles situated in planes parallel
with the plane }}$xOy$\textbf{, }\textit{\textbf{having the centers situated
on the axis }}$Oz$\textbf{, }\textit{\textbf{and the radii }}$r_{0}\sin
\theta _{0}$\textbf{, }\textit{\textbf{where }}$r_{0}>0$\textbf{\ }\textit{%
\textbf{is a solution of the equation }}$f^{\prime }(r)=0$. By comparing, we
recall that in the case of cylinder solutions we obtained \textit{the
generators of the right circular cylinders}%
\begin{equation*}
\left( x^{1}\right) ^{2}+\left( x^{2}\right) ^{2}=\rho _{0}^{2},
\end{equation*}%
where $\rho _{0}>0$\ is a solution of the same equation.
\end{enumerate}

\textbf{Case 2: }Let us consider that $\dot{\theta}=0$ and $v\neq 0$. We
will look for solutions of the form%
\begin{equation*}
c(\phi )=\left( r=r(\phi ),\text{ }\theta =\text{constant},\text{ }\phi
:=t\in \lbrack 0,2\pi ]\right) \Rightarrow
\end{equation*}%
\begin{equation*}
\Rightarrow \dot{c}(\phi )=\left( \dot{r}=\frac{dr}{d\phi },\text{ }\dot{%
\theta}=0,\text{ }\dot{\phi}=1\right) .
\end{equation*}%
These conditions imply%
\begin{equation*}
v^{2}=\left( \frac{dr}{dt}\right) ^{2}+r^{2}\sin ^{2}\theta .
\end{equation*}

In this case the anisotropic equations of motion reduce to%
\begin{equation*}
\begin{array}{c}
\sin \theta \cos t\left( \dfrac{d^{2}r}{dt^{2}}\right) -2\sin \theta \sin
t\left( \dfrac{dr}{dt}\right) -r\sin \theta \cos t+\medskip \\ 
+\dfrac{4ff^{\prime }v^{2}\left( 1+3f^{2}v^{2}\right) }{\left(
1+2f^{2}v^{2}\right) \left( 1+6f^{2}v^{2}\right) }\left( \sin \theta \cos
t\left( \dfrac{dr}{dt}\right) ^{2}-r\sin \theta \sin t\left( \dfrac{dr}{dt}%
\right) \right) -\medskip \\ 
-\dfrac{ff^{\prime }v^{4}}{1+2f^{2}v^{2}}\sin \theta \cos t=0,%
\end{array}%
\end{equation*}%
\begin{equation*}
\begin{array}{c}
\sin \theta \sin t\left( \dfrac{d^{2}r}{dt^{2}}\right) +2\sin \theta \cos
t\left( \dfrac{dr}{dt}\right) -r\sin \theta \sin t+\medskip \\ 
+\dfrac{4ff^{\prime }v^{2}\left( 1+3f^{2}v^{2}\right) }{\left(
1+2f^{2}v^{2}\right) \left( 1+6f^{2}v^{2}\right) }\left( \sin \theta \sin
t\left( \dfrac{dr}{dt}\right) ^{2}+r\sin \theta \cos t\left( \dfrac{dr}{dt}%
\right) \right) -\medskip \\ 
-\dfrac{ff^{\prime }v^{4}}{1+2f^{2}v^{2}}\sin \theta \sin t=0,%
\end{array}%
\end{equation*}%
\begin{equation*}
\begin{array}{c}
\cos \theta \left( \dfrac{d^{2}r}{dt^{2}}\right) +\dfrac{4ff^{\prime
}v^{2}\left( 1+3f^{2}v^{2}\right) }{\left( 1+2f^{2}v^{2}\right) \left(
1+6f^{2}v^{2}\right) }\cos \theta \left( \dfrac{dr}{dt}\right) ^{2}-\medskip
\\ 
-\dfrac{ff^{\prime }v^{4}}{1+2f^{2}v^{2}}\cos \theta =0.%
\end{array}%
\end{equation*}

\begin{enumerate}
\item If $\cos \theta =0\Leftrightarrow \theta =\pi /2$, then the equations
of motion become%
\begin{equation}
\begin{array}{c}
\cos t\left( \dfrac{d^{2}r}{dt^{2}}\right) -2\sin t\left( \dfrac{dr}{dt}%
\right) -r\cos t-\dfrac{ff^{\prime }v^{4}}{1+2f^{2}v^{2}}\cos t+\medskip \\ 
+\dfrac{4ff^{\prime }v^{2}\left( 1+3f^{2}v^{2}\right) }{\left(
1+2f^{2}v^{2}\right) \left( 1+6f^{2}v^{2}\right) }\cdot \medskip \\ 
\cdot \left( \cos t\left( \dfrac{dr}{dt}\right) ^{2}-r\sin t\left( \dfrac{dr%
}{dt}\right) \right) =0,%
\end{array}
\label{eqms-1}
\end{equation}%
\begin{equation}
\begin{array}{c}
\sin t\left( \dfrac{d^{2}r}{dt^{2}}\right) +2\cos t\left( \dfrac{dr}{dt}%
\right) -r\sin t-\dfrac{ff^{\prime }v^{4}}{1+2f^{2}v^{2}}\sin t+\medskip \\ 
+\dfrac{4ff^{\prime }v^{2}\left( 1+3f^{2}v^{2}\right) }{\left(
1+2f^{2}v^{2}\right) \left( 1+6f^{2}v^{2}\right) }\cdot \medskip \\ 
\cdot \left( \sin t\left( \dfrac{dr}{dt}\right) ^{2}+r\cos t\left( \dfrac{dr%
}{dt}\right) \right) =0,%
\end{array}
\label{eqms-2}
\end{equation}%
where%
\begin{equation*}
v^{2}=\left( \frac{dr}{dt}\right) ^{2}+r^{2}.
\end{equation*}%
Multiplying the first equation by $\left( -\sin t\right) $, and the second
by $\left( \cos t\right) $, by summing we get the equation:%
\begin{equation}
\left( \dfrac{dr}{dt}\right) \left( 1+\dfrac{2rff^{\prime }v^{2}\left(
1+3f^{2}v^{2}\right) }{\left( 1+2f^{2}v^{2}\right) \left(
1+6f^{2}v^{2}\right) }\right) =0.  \label{eqs-last}
\end{equation}

\begin{enumerate}
\item If $dr/dt=0$, then the equations of motion become%
\begin{equation*}
\left\{ 
\begin{array}{l}
r\cos t+\dfrac{ff^{\prime }v^{4}}{1+2f^{2}v^{2}}\cos t=0\medskip \\ 
r\sin t+\dfrac{ff^{\prime }v^{4}}{1+2f^{2}v^{2}}\sin t=0%
\end{array}%
\right. \Leftrightarrow
\end{equation*}%
\begin{equation*}
\Leftrightarrow r+\dfrac{ff^{\prime }v^{4}}{1+2f^{2}v^{2}}=0\Leftrightarrow
\end{equation*}%
\begin{equation}
\Leftrightarrow 1+2f^{2}r^{2}+ff^{\prime }r^{3}=0,  \label{eqs-2}
\end{equation}%
where $v^{2}=r^{2}$. The corresponding solutions are%
\begin{equation*}
c(t)=\left( r=r_{0}>0,\text{ }\theta =\frac{\pi }{2},\text{ }\phi =t\right)
\Rightarrow
\end{equation*}%
\begin{equation*}
\Rightarrow 
\begin{array}{ccc}
x^{1}(t)=r_{0}\cos t, & x^{2}(t)=r_{0}\sin t, & x^{3}(t)=0.%
\end{array}%
\end{equation*}%
It follows that in this case \textit{the particles move on \textbf{circles
situated in the plane }}$xOy$\textbf{, }\textit{\textbf{having the centers
in origin}}\textbf{, }\textit{\textbf{and the radii }}$\rho _{0}>0$\textbf{, 
}\textit{\textbf{where }}$\rho _{0}$\textbf{\ }\textit{\textbf{is a solution
of the equation (\ref{eqs-2}). }}Note that these trajectories also appear in
the case of cylindric solutions (see again the preceding Case 3. // Subcase
1.).

\item If $dr/dt\neq 0$, then we get 
\begin{equation*}
1+\dfrac{2rff^{\prime }v^{2}\left( 1+3f^{2}v^{2}\right) }{\left(
1+2f^{2}v^{2}\right) \left( 1+6f^{2}v^{2}\right) }=0.
\end{equation*}%
The preceding equation has solution if and only if%
\begin{equation*}
2f+rf^{\prime }\in \left( -\frac{1+4r^{2}f^{2}}{2r^{2}f\left(
1+3r^{2}f^{2}\right) },0\right) ,
\end{equation*}%
and the solutions are%
\begin{equation}
\frac{dr}{dt}=\pm \sqrt{\frac{-4f^{2}-rff^{\prime }-\sqrt{\Delta ^{\prime }}%
}{6f^{3}\left( 2f+rf^{\prime }\right) }-r^{2}}:=F\left( r\right) ,
\label{eqs-3}
\end{equation}%
where $\Delta ^{\prime }=4f^{4}+r^{2}f^{2}\left( f^{\prime }\right)
^{2}+2rf^{3}f^{\prime }$. Imposing as a solution of (\ref{eqs-3}) to be also
a solution for equations (\ref{eqms-1}) and (\ref{eqms-2}), we get%
\begin{equation*}
FF^{\prime }-r-\frac{2F^{2}}{r}=\frac{ff^{\prime }\left( F^{2}+r^{2}\right)
^{2}}{1+2f^{2}\left( F^{2}+r^{2}\right) }\Rightarrow
\end{equation*}%
\begin{equation*}
\Rightarrow r=\text{constant (contradiction!)}.
\end{equation*}%
So, the subcase $dr/dt\neq 0$ does not provide us any solution.
\end{enumerate}

\item If $\cos \theta \neq 0$ and $\sin \theta \neq 0$, then the equations
of motion reduce to the equations (\ref{eqms-1}), (\ref{eqms-2}) and%
\begin{equation}
\begin{array}{c}
\left( \dfrac{d^{2}r}{dt^{2}}\right) +\dfrac{4ff^{\prime }v^{2}\left(
1+3f^{2}v^{2}\right) }{\left( 1+2f^{2}v^{2}\right) \left(
1+6f^{2}v^{2}\right) }\left( \dfrac{dr}{dt}\right) ^{2}-\medskip \\ 
-\dfrac{ff^{\prime }v^{4}}{1+2f^{2}v^{2}}=0.%
\end{array}
\label{eqms-theta}
\end{equation}

Using the equation (\ref{eqms-theta}), we deduce that the equations (\ref%
{eqms-1}) and (\ref{eqms-2}) simplify as%
\begin{equation*}
-2\sin t\left( \dfrac{dr}{dt}\right) -r\cos t-\dfrac{4rff^{\prime
}v^{2}\left( 1+3f^{2}v^{2}\right) }{\left( 1+2f^{2}v^{2}\right) \left(
1+6f^{2}v^{2}\right) }\sin t\left( \dfrac{dr}{dt}\right) =0,
\end{equation*}%
\begin{equation*}
2\cos t\left( \dfrac{dr}{dt}\right) -r\sin t+\dfrac{4rff^{\prime
}v^{2}\left( 1+3f^{2}v^{2}\right) }{\left( 1+2f^{2}v^{2}\right) \left(
1+6f^{2}v^{2}\right) }\cos t\left( \dfrac{dr}{dt}\right) =0.
\end{equation*}%
Multiplying the first equation from above by $\left( \cos t\right) $, and
the second by $\left( \sin t\right) $, we deduce that this system of
equations has no solution.

\item If $\sin \theta =0$, then the equations of motion are equivalent with
differential equation (\ref{eqms-theta}), where%
\begin{equation*}
v^{2}=\left( \frac{dr}{dt}\right) ^{2}.
\end{equation*}%
To integrate the equation (\ref{eqms-theta}) seems to be very difficult but
it is obvious that \textit{the solutions of the equations of motion are 
\textbf{segments of the axis }}$Oz$, whose lengths are determined by the
images of the solutions $r(t)$ of the equation (\ref{eqms-theta}).
\end{enumerate}

\textbf{Case 3: }Let us consider that $\dot{\phi}=0$ and $v\neq 0$. We will
look for solutions of the form%
\begin{equation*}
c(\theta )=\left( r=r(\theta ),\text{ }\theta :=t\in \lbrack 0,\pi ],\text{ }%
\phi =\text{constant}\right) \Rightarrow
\end{equation*}%
\begin{equation*}
\Rightarrow \dot{c}(\theta )=\left( \dot{r}=\frac{dr}{d\theta },\text{ }\dot{%
\theta}=1,\text{ }\dot{\phi}=0\right) .
\end{equation*}%
These conditions imply%
\begin{equation*}
v^{2}=\left( \frac{dr}{dt}\right) ^{2}+r^{2}.
\end{equation*}

In this case the anisotropic equations of motion reduce to%
\begin{equation*}
\begin{array}{c}
\sin t\cos \phi \left( \dfrac{d^{2}r}{dt^{2}}\right) +2\cos t\cos \phi
\left( \dfrac{dr}{dt}\right) -r\sin t\cos \phi +\medskip \\ 
+\dfrac{4ff^{\prime }v^{2}\left( 1+3f^{2}v^{2}\right) }{\left(
1+2f^{2}v^{2}\right) \left( 1+6f^{2}v^{2}\right) }\left( \sin t\cos \phi
\left( \dfrac{dr}{dt}\right) ^{2}+r\cos t\cos \phi \left( \dfrac{dr}{dt}%
\right) \right) -\medskip \\ 
-\dfrac{ff^{\prime }v^{4}}{1+2f^{2}v^{2}}\sin t\cos \phi =0,%
\end{array}%
\end{equation*}%
\begin{equation*}
\begin{array}{c}
\sin t\sin \phi \left( \dfrac{d^{2}r}{dt^{2}}\right) +2\cos t\sin \phi
\left( \dfrac{dr}{dt}\right) -r\sin t\sin \phi +\medskip \\ 
+\dfrac{4ff^{\prime }v^{2}\left( 1+3f^{2}v^{2}\right) }{\left(
1+2f^{2}v^{2}\right) \left( 1+6f^{2}v^{2}\right) }\left( \sin t\sin \phi
\left( \dfrac{dr}{dt}\right) ^{2}+r\cos t\sin \phi \left( \dfrac{dr}{dt}%
\right) \right) -\medskip \\ 
-\dfrac{ff^{\prime }v^{4}}{1+2f^{2}v^{2}}\sin t\sin \phi =0,%
\end{array}%
\end{equation*}%
\begin{equation*}
\begin{array}{c}
\cos t\left( \dfrac{d^{2}r}{dt^{2}}\right) -2\sin t\left( \dfrac{dr}{dt}%
\right) -r\cos t+\medskip \\ 
+\dfrac{4ff^{\prime }v^{2}\left( 1+3f^{2}v^{2}\right) }{\left(
1+2f^{2}v^{2}\right) \left( 1+6f^{2}v^{2}\right) }\left( \cos t\left( \dfrac{%
dr}{dt}\right) ^{2}-r\sin t\left( \dfrac{dr}{dt}\right) \right) -\medskip \\ 
-\dfrac{ff^{\prime }v^{4}}{1+2f^{2}v^{2}}\cos t=0.%
\end{array}%
\end{equation*}

Multiplying the first equation by $\left( \cos \phi \cos t\right) $, the
second by $\left( \sin \phi \cos t\right) $, and the third by $\left( -\sin
t\right) $, by summing we get again the equation (\ref{eqs-last}). But, this
case was previously treated, and the corresponding solutions are%
\begin{equation*}
c(t)=\left( r=r_{0}>0,\text{ }\theta =t,\text{ }\phi =\phi _{0}\in \lbrack
0,2\pi ]\right) \Rightarrow
\end{equation*}%
\begin{equation*}
\Rightarrow 
\begin{array}{ccc}
x^{1}(t)=\left( r_{0}\cos \phi _{0}\right) \sin t, & x^{2}(t)=\left(
r_{0}\sin \phi _{0}\right) \sin t, & x^{3}(t)=r_{0}\cos t,%
\end{array}%
\end{equation*}%
where $r_{0}>0$\textbf{\ }is a solution of the equation (\ref{eqs-2}). It
follows that in this case \textit{the trajectories are the following \textbf{%
circles:}}%
\begin{equation*}
\left\{ 
\begin{array}{l}
\left( x^{1}\right) ^{2}+\left( x^{2}\right) ^{2}+\left( x^{3}\right)
^{2}=r_{0}^{2}\medskip \\ 
\left( \sin \phi _{0}\right) x^{1}-\left( \cos \phi _{0}\right) x^{2}=0.%
\end{array}%
\right.
\end{equation*}

\section{Conclusion}

At the end of this paper, we consider it is important to underline some
geometrical differences between the classical non-relativistic dynamics
(determined only by the Euclidian metric $\delta _{ij}$) and the present
anisotropic non-relativistic extended dynamics (determined by the perturbed
Lagrangian (\ref{Lagragian-ARO})). For instance, we recall that in the first
case we have only a Riemannian curvature (determined by the Levi-Civita
connection), while in the second case we have more Lagrangian curvatures
(determined by the more complicated Cartan connection). Moreover, in the
first case we live in a flat space (its Riemannian curvature is zero), while
in the second case we work in a curved spaces (its Lagrangian curvatures
being non-zero). At the same time, note that the corresponding geodesics are
in these different cases:

\begin{enumerate}
\item any straight lines in space, in the first case;

\item only certain circular helices or circles situated in some specific
planes, together with certain specific straight lines which are parallel
with the axis $Oz$, in the second case. It is also important to note that,
in this second anisotropic situation, the preceding geometrical geodesics do
not represent the general solutions of the geodesic equations, but only some
particular cases of them. To find the general solution of these geodesic
equations can be considered as an open problem for a research study in this
geometric-physical domain.
\end{enumerate}

\noindent\textbf{Acknowledgements. }The present work was developed under the
auspices of Grant 1196/2012 - BRFFR-RA F12RA-002, within the cooperation
framework between Romanian Academy and Belarusian Republican Foundation for
Fundamental Research.

\end{document}